\documentclass[journal,letterpaper,twocolumn,twoside,nofonttune]{IEEEtran}
\usepackage{etex}
\normalsize
\usepackage[T1]{fontenc}

\usepackage{amsmath,amssymb,amsfonts}
\usepackage{mathrsfs}
\usepackage{mathabx}
\usepackage{amsbsy}
\usepackage{graphicx}
\usepackage{graphics}
 \usepackage{epstopdf}
\usepackage{algorithm}
\usepackage{algorithmic}
\usepackage{subfigure}
\usepackage{cite}
\usepackage{multirow}
\usepackage{ctable}
\usepackage{caption}
\usepackage{setspace}
\usepackage{mathtools}
\usepackage{lipsum}
\usepackage[algo2e]{algorithm2e}
\usepackage{diagbox}

\title{\Huge$\,$\\[-2.75ex]
{Analog Lagrange Coded Computing}\\[0.50ex]}

\author{\large%
	Mahdi Soleymani, Hessam Mahdavifar, and A. Salman Avestimehr
	\vspace{-.25in}
	\thanks{This work was supported by the National Science Foundation
		under grants CCF--1763348, CCF--1909771, CCF--1941633.}
	\thanks{M.\ Soleymani and H.\ Mahdavifar are with the Department of Electrical Engineering and Computer Science, University of Michigan, Ann Arbor, MI 48104 (email: mahdy@umich.edu and hessam@umich.edu).}
	\thanks{A.\ Salman Avestimehr is with the Department of Electrical
		Engineering, University of Southern California, Los Angeles, CA 90089 USA
		(e-mail: avestimehr@ee.usc.edu).}
}
\newtheorem{exmp}{Example}[section]
\newtheorem{theorem}{{Theorem}}
\newtheorem{lemma}[theorem]{{Lemma}}

\newtheorem{corollary}[theorem]{{Corollary}}

\newtheorem{remark}{{\textbf{Remark}}}
\AfterEndEnvironment{remark}{\noindent\ignorespaces}
\newcommand{\cF}{{\cal F}}

\DeclareMathAlphabet{\mathbfsl}{OT1}{ppl}{b}{it} %{OT1}{cmr}{bx}{it}

\newcommand{\bA}{\mathbfsl{A}} 
\newcommand{\bB}{\mathbfsl{B}}

\newcommand{\bH}{\mathbfsl{H}}
\newcommand{\bI}{\mathbfsl{I}}

\newcommand{\bL}{\mathbfsl{L}}
 
\newcommand{\bN}{\mathbfsl{N}}

\newcommand{\bV}{\mathbfsl{V}}
\newcommand{\bW}{\mathbfsl{W}} 
\newcommand{\bX}{\mathbfsl{X}}
\newcommand{\bY}{\mathbfsl{Y}}

\newcommand{\ba}{\mathbfsl{a}}

\newcommand{\bv}{\mathbfsl{v}}
 
\newcommand{\bx}{\mathbfsl{x}}
\newcommand{\by}{\mathbfsl{y}}

\newcommand{\bn}{\mathbfsl{n}} 
\newcommand{\bff}{\mathbfsl{f}}
\newcommand{\bSig}{\mathbf{\Sigma}} 

\newcommand{\ti}[1]{\tilde{#1}}
\newcommand*{\rom}[1]{\expandafter\romannumeral #1}
\newcommand{\tL}{\tilde{\bL}}
\makeatletter
\newcommand{\AlignFootnote}[1]{%
  \ifmeasuring@
  \else
    \iffirstchoice@
      \footnote{#1}%
    \fi
  \fi}
\makeatother

\newcommand{\norm}[1]{\left\lVert#1\right\rVert}
\usepackage[utf8]{inputenc}

\newcommand{\be}[1]{\begin{equation}\label{#1}}
\newcommand{\ee}{\end{equation}} 
\newcommand{\eq}[1]{(\ref{#1})}

\renewcommand{\leq}{\leqslant}

\renewcommand{\Bbb}{\mathbb}
\newcommand{\C}{{\Bbb C}} 

\newcommand{\R}{{\Bbb R}} 
\newcommand{\Z}{{\Bbb Z}}
\newcommand{\F}{{\Bbb F}}

\newcommand{\Tref}[1]{Theo\-rem\,\ref{#1}}

\newcommand{\Lref}[1]{Lem\-ma\,\ref{#1}}
\newcommand{\Cref}[1]{Co\-ro\-lla\-ry\,\ref{#1}}
\newcommand{\Rref}[1]{Re\-mark\,\ref{#1}}

\newcommand{\deff}{\mbox{$\stackrel{\rm def}{=}$}}

\newcommand{\erel}{e_{\text{rel}}}
\begin{document}

\maketitle

\begin{abstract}
A distributed computing scenario is considered, where the computational power of a set of worker nodes is used to perform a certain computation task over a dataset that is dispersed among the workers. Lagrange coded computing (LCC), proposed by Yu et al., leverages the well-known Lagrange polynomial to perform polynomial evaluation of the dataset in such a scenario in an efficient parallel fashion while keeping the privacy of data amidst possible collusion of workers. This solution relies on quantizing the data into a finite field, so that Shamir's secret sharing, as one of its main building blocks, can be employed. Such a solution, however, is not properly scalable with the size of dataset, mainly due to computation overflows. To address such a critical issue, we propose a novel extension of LCC to the analog domain, referred to as analog LCC (ALCC). All the operations in the proposed ALCC protocol are done over the infinite fields of $\R$/$\C$ but for practical implementations floating-point numbers are used. We characterize the \emph{privacy} of data in ALCC, against any subset of colluding workers up to a certain size, in terms of the distinguishing security (DS) and the mutual information security (MIS) metrics. Also, the \emph{accuracy} of outcome is characterized in a practical setting assuming operations are performed using floating-point numbers. Consequently, a fundamental trade-off between the accuracy of the outcome of ALCC and its privacy level is observed and is numerically evaluated. Moreover, we implement the proposed scheme to perform matrix-matrix multiplication over a batch of matrices. It is observed that ALCC is superior compared to the state-of-the-art LCC, implemented using fixed-point numbers, assuming both schemes use an equal number of bits to represent data symbols. 
\end{abstract}

\begin{IEEEkeywords} 
Coded computing, privacy-preserving computing, analog coding
\end{IEEEkeywords}

%==============================================================================%
%                                                                              %
%   1. INTRODUCTION                                                            %
%                                                                              %
%==============================================================================%
\section{Introduction} 
\label{sec:Introduction}
There has been a growing interest in recent years towards performing computational tasks across networks of computational worker nodes by utilizing their computational power in a parallel fashion \cite{agrawal2000privacy, nikolaenko2013privacy,lee2017speeding,yu2019lagrange,li2020coded}. Computations over massive datasets need to be carried out at an unprecedented scale that entails solutions scalable with the size of datasets associated with a wide range of problems including machine learning \cite{abadi2016tensorflow}, optimization \cite{rabbat2004distributed}, etc. A well-established network architecture to perform such tasks in a distributed fashion consists of a \emph{master} node together with a set of worker nodes having communication links only with the master node \cite{lee2017speeding,yu2019lagrange}. In such systems, a dataset is dispersed among the servers across the network to perform a certain computational task over the dataset. The master node then aggregates the results in order to recover the desired outcome, e.g., the output of a certain function over the dataset.  

 Dispersing data across a network gives rise to several fundamental challenges in practice. One of the major concerns in such systems is to keep the data \emph{private} as the computational tasks often involve sensitive data such as patients recordings, financial transactions, etc \cite{wright2004privacy, kantarcioglu2004privacy, clifton2002tools}. The worker nodes are often assumed to be \emph{honest-but-curious}, i.e., they do not deviate from the protocol but may accumulate the shares of data they receive and try to deduce information about the data. In such settings, the challenge is to utilize the computational power of the nodes while ensuring that \emph{almost} no information about the dataset is revealed to them. Furthermore, this restriction is often extended to preserving the privacy of data against any subset of colluding nodes up to a certain size. 
 
 Several security metrics are considered in different contexts to measure privacy/security of data. This includes semantic security (SS) and distinguishing security (DS) in the cryptography literature \cite{goldwasser1984probabilistic}, mutual information security (MIS) in communication settings \cite{shannon1949communication}, differential security in machine learning \cite{dinur2003revealing}, etc. From the information-theoretic perspective, the \emph{perfect} privacy condition in a distributed computation setting is that \emph{no information} is leaked about the dataset to any of the worker nodes/subsets of colluding worker nodes up to a certain size. To this end, Shamir's seminal secret sharing scheme is the main building block in protocols providing \emph{perfect} privacy in these settings \cite{shamir1979share}. In such protocols, the data symbols are always assumed to be elements of a finite field $\F_p$ leading to perfect privacy guarantees. However, this often comes at the expense of substantial accuracy losses due to fixed-point representation of the data and computation overflows. Especially, this becomes a major barrier in scalability of such protocols with respect to the dataset size. 
 
The seminal Shamir's secret sharing scheme and its various versions are often used to provide information-theoretic security for data, referred to as a secret, while distributing it among a set of servers/users \cite{shamir1979share}. Also, Shamir's scheme serves as the backbone of most of the existing schemes on privacy-preserving multiparty computing such as the celebrated BGW scheme \cite{ben1988completeness}. In $(n,k)$ Shamir's secret sharing scheme, the secret, which  is regarded as an element of a finite field, is encoded to a polynomial of degree $k-1$ whose constant coefficient is the secret and all other coefficients are picked uniformly at random from the field. The shares are then evaluations of the polynomial at $n$ distinct points. The secret can be uniquely decoded when at least $k$ number of shares are available while no information is revealed about the secret otherwise. In order to employ Shamir-based distributed computing protocols the data is quantized and then mapped to a finite field at the beginning. This leads to a possibly substantial loss in the accuracy of the computation results mainly due to computation overflows when the dataset is \textit{large}. In order to overcome this issue, an analog counterpart of Shamir's scheme is recently proposed in  \cite{soleymani2020privacypreserving} and is then utilized to perform a learning task when the data is provided using floating-point numbers. Lagrange coded computing (LCC) \cite{yu2019lagrange} provides a framework to efficiently perform distributed computation over a batch of data in a parallel fashion. It can be utilized to provide privacy-preserving machine learning schemes.  Similar to Shamir's scheme, in LCC, the data is assumed to be an element of a finite field and the secret/data is encoded to a certain polynomial, called Lagrange interpolation polynomial. Hence, the loss in accuracy due to overflow in computations  is inherited to LCC as well.

\subsection{Our contribution }
 
In this paper, we propose a framework to extend the privacy-preserving LCC scheme to the analog domain and refer to it as analog LCC (ALCC). It is assumed that all the worker nodes are honest-but-curious. All the operations in the proposed scheme are done over the infinite fields of $\R$/$\C$ but for practical implementations floating-point numbers are used. The proposed ALCC protocol enables \emph{privately} evaluating a polynomial function over a batch of real/complex-valued dataset in parallel. We characterize the performance of the scheme in terms of the \emph{accuracy} of its outcome, when operations are performed using standard floating-point numbers, and the \emph{privacy} of data in terms of the DS and MIS metrics when any subset of worker nodes up to a certain size can collude. It is shown how various parameters of the ALCC protocol, including parameters associated with Lagrange monomials as well as evaluation points in the complex plane, can be carefully picked in order to provide closed-form bounds for the performance of the protocol from both the privacy and the accuracy perspectives. Furthermore,  a fundamental trade-off between the accuracy of the outcome of ALCC and its privacy level is observed and is numerically evaluated, in terms of various parameters of the scheme, when the scheme is implemented using the floating-point numbers. In a related work, we show that accuracy-privacy trade-offs arise in distributed computing in the analog domain by tuning the noise variance in the underlying protocol \cite{soleymani2020privacypreserving}. However, the main distinction of the current paper is to illustrate that, even for a fixed noise variance, the choice of certain parameters of Lagrange monomials in ALCC leads to a new trade-off between accuracy and privacy which is specific to ALCC. Hence, one has to carefully pick these parameters apart from the noise variance in order to avoid unnecessarily compromising  accuracy/privacy in practice. This is a new fundamental trade-off that does not have a counterpart in either analog adaptations of Shamir's scheme, e.g., \cite{soleymani2020privacypreserving}, or LCC with fixed-point implementation over finite fields \cite{yu2019lagrange}. Also, it is numerically illustrated that ALCC scales better with the number of representation bits considered in the floating-point implementation compared to LCC. Moreover, experiments are shown in which the proposed protocol is implemented to perform matrix-matrix multiplication over a batch of matrices. The results indicate the superiority of the proposed ALCC compared to the state-of-the-art LCC implemented using fixed-point numbers assuming both schemes use an equal number of bits to represent each data symbol. 

Note that LCC simultaneously
provides resiliency against stragglers, security against malicious workers or workers with erroneous returned results, and privacy of the dataset \cite{yu2019lagrange}. Here, we are mainly concerned with the privacy of dataset in the analog domain. Our analysis of the proposed scheme also takes into account the issue with slow/unresponsive nodes, also referred to as stragglers. However, the issue with malicious workers is left for future work.   

\subsection{Related work}

Privacy-preserving distributed computing protocols have been recently studied in a wide range of scenarios to fulfill specific privacy requirements \cite{dahl2018private,barak2019secure,so2019codedprivateml,kumar2019cryptflow,so2020turbo}. Furthermore, secure matrix-matrix multiplication, as one of the main building blocks for various machine learning algorithms, has been extensively studied in the literature \cite{yu2020entangled,aliasgari2020private,d2020gasp,bitar2019private,nodehi2019secure,9229375}. Also, Lagrange coded computing \cite{yu2019lagrange} and its variations \cite{raviv2019private,fahim2019lagrange} provide a framework for evaluating a given polynomial function over a dataset with perfect privacy \cite{yu2019lagrange}. Such protocols have been recently adopted to perform various machine learning tasks. Recently, it is shown that LCC can be employed to break the aggregation barrier in secured federated learning \cite{so2020turbo}. However, these prior works often regard data as elements of a finite filed.  %However, these prior works are often based upon Shamir's secret sharing scheme in the finite field and its variations. 
As a result, they suffer from scalability issues, as discussed earlier. By enabling privacy in the analog domain, ALCC  provides a framework to perform several large-scale tasks, e.g., secure aggregation in federated learning  \cite{so2020turbo}, more efficiently in practice.    

There is also another line of work on privacy-preserving machine learning that utilizes off-the-shelf multi party computation (MPC) protocols\cite{yao1982protocols,ben2019completeness}  to train a model over distributed datasets \cite{nikolaenko2013privacy,gascon2017privacy,mohassel2017secureml,agrawal2000privacy,dahl2018private,chen2019secure}. In these MPC-based machine learning problems, often more than one client are assumed that aim at learning  model parameters collaboratively without sharing sensitive data with each other and  worker nodes. On the other hand, in (A)LCC, it is assumed that all the data is present in one central node/client, called the master node, that utilizes the computational power of worker nodes for speed up while keeping the data private from the workers. In other words, the MPC-based schemes mainly concern with the privacy of sensitive datasets over which a model is trained in a fully distributed fashion, while (A)LCC-based methods provide a framework for privacy-preserving machine learning in which the dataset is offloaded to a cloud-computing environment to gain speed up \cite{so2019codedprivateml}. However, MPC-based ML schemes are also adopted in the case with one central client in the literature \cite{mohassel2017secureml, so2019codedprivateml}. In this approach, only a few number of worker nodes are often considered mainly due to inefficiency of underlying MPC schemes. In a recent work, a fully distributed implementation of LCC is introduced in \cite{so2020scalable}, which is then utilized to train a linear regression model over distributed datasets without considering a \emph{central entity} and is significantly faster than  MPC-based methods. In general, (A)LCC schemes reduce the amount of randomness needed in data encoding and have less storage overhead as well as computation complexity. Moreover, no communication is needed between worker nodes in (A)LCC, a factor that contributes the most to the inefficiency of MPC-based ML in practical systems. There is also a line of work concerning with floating-point implementation of MPC protocols \cite{setty2012taking,aliasgari2013secure,catrinatowards} which requires significantly more rounds of communications and computations compared to the conventional MPC protocols with fixed-point implementation. As a result, the inefficiency of such protocols poses a major difficulty in their implementation as well.

 Another line of work on performing computations over real-valued data is considered in \cite{fahim2019numerically,ramamoorthy2019numerically,das2019distributed,jamali2019coded,charalambides2020numerically}. In these works, the coded distributed computing schemes are adapted to the analog domain by addressing the numerical stability issues arising in the inversion of underlying Vandermonde matrices. However, the privacy constraints are not considered in these works. In this paper, however, our main focus is on providing privacy-preserving schemes in the analog domain. Also, codes in the analog domain have been recently studied in the context of block codes \cite{roth2020analog} as well as subspace codes \cite{soleymani2019analog} for analog error correction. However, secret sharing and privacy-preserving computation in the analog domain are not discussed in these works.

The rest of this paper is organized as follows.
In Section\,\ref{sec:System Model}, the system model is discussed and the proposed protocol is described. The accuracy of the protocol is analyzed in Section\,\ref{sec:Accuracy}. In Section\,\ref{Sec:privacy} the privacy level of data in ALCC is characterized in terms of two well-known notions of security. Various experimental results are provided in Section\,\ref{sec:experiments}. Finally, the paper is concluded in Section\,\ref{sec:conclusion}.

%==============================================================================%
%                                                                             %
%   1.SYSTEM MODEL                                                           %
%                                                                              %
%==============================================================================%
\section{System Model} 
\label{sec:System Model}

Consider a dataset $\bX=(\bX_1, \hdots,\bX_k)$ with $\bX_i \in \R^{m\times n}$ for all $i \in [k]$, where $[k]$ denotes $\{1,2,\dots,k\}$. Each entry of $\bX_i$'s is assumed to be an instance of a continuous random variable with the range $[-r,r]$.  No further assumptions is made on the probability distribution of the entries of $\bX_i$'s.

We consider the problem of evaluating a polynomial $f:\R^{m \times n} \rightarrow \R^{u\times h}$ over the dataset $\bX$ in a distributed fashion while keeping the \textit{privacy} of $\bX$. More specifically, we say $f(\cdot)$ is a $D$-degree polynomial function if all entries of the output matrix are multivariate polynomial functions of the entries of the input with total degree at most $D$, i.e., $\bY=f(\bX)$ implies that
\be{polynomial-def}
y_{ij}=f_{ij}(x_{11},x_{12},\cdots, x_{mn}),
\ee
where $y_{ij}$ is the $(i,j)$ entry of $\bY$, for $i \in [u]$ and $j \in [h]$, $x_{lk}$ is the $(l,k)$ entry of $\bX$, for $l \in [m]$ and $k \in [n]$, and, $f_{ij}$ is a multivariate polynomial of total degree at most $D$. We may write the right hand side of \eqref{polynomial-def} as $f_{ij}(\bX)$, or simply $f_{ij}$ when the argument is clear from the context, throughout the rest of the paper. The distributed computing setup consists of a master node and $N$ worker nodes/parties. It is assumed that there is no communication link between the parties. More specifically, in this setup, the goal of the master node is to compute $f(\bX_i)$ for all $i \in [k]$, where $f$ is a degree-$D$ polynomial, using the computational power of the parties. This is done in such a way that the dataset is kept \emph{private} from the parties assuming up to a certain threshold, denoted by $t$, of them can collude. The notion of privacy in the analog domain will be clarified in Section \ref{Sec:privacy}. Note that this setup is similar to the one considered for LCC in \cite{yu2019lagrange} with the main difference that in \cite{yu2019lagrange} the dataset and all the computations are assumed to be over a finite field. More specifically, our problem setup can be regarded as an extension of the problem setup considered in \cite{yu2019lagrange} to the analog domain.

Next we discuss the encoding process in analog Lagrange coded computing (ALCC), i.e., how to encode the dataset $\bX$ into the shares distributed to the worker nodes. Let $\bW$ denote $(\bX_1, \hdots, \bX_k, \bN_1, \hdots, \bN_t)$ where $\bN_i$'s are $m\times n$ random matrices with i.i.d. entries drawn from a zero-mean circular symmetric complex Gaussian distribution with standard deviation $\frac{\sigma_n}{\sqrt{t}}$, denoted by $\mathcal{N}(0,\frac{\sigma_n^2}{t})$, with $t$ being the maximum number of colluding parties. Note that a zero-mean circular symmetric  complex Gaussian random variable (RV) with variance $\sigma^2$ consists of two i.i.d. zero-mean Gaussian RV's with variance $\frac{\sigma^2}{2}$ as its real and imaginary part. Let $\gamma$ and $\omega$ denote the $N$-th and the $(k+t)$-th root of unity, respectively. In other words, $\gamma = \exp(\frac{2\pi i}{N})$ and $\omega = \exp(\frac{2\pi i}{k+t})$, where $i^2=-1$. In ALCC, the Lagrange polynomial is constructed as 
\be{Lagrange_polynomial} 
u(z)=\sum_{j=1}^k \bX_j l_j(z)+\sum_{j=k+1}^{k+t} \bN_{j-k} l_j(z)=\sum_{j=1}^{k+t} \bW_i l_j(z),
\ee
where $l_j(\cdot)$'s are Lagrange monomials defined as 
\be{Lagrange_monomials} 
l_j(z)=\prod_{l\in [k+t]\setminus j} \frac{z-\beta_l}{\beta_j-\beta_l},
\ee 
for all $j\in [k+t]$. Furthermore, the parameters $\beta_j$'s are picked to be equally spaced on the circle of radius $\beta$ centered around $0$ in the complex plane, for some $\beta \in \R$, i.e.,
\be{beta}
\beta_j=\beta \omega^{j-1}.
\ee
The shares of encoded dataset to be distributed to the worker nodes consist of the evaluation of $u(z)$ over the $N$-th roots of unity in the complex plane, i.e.,
\be{shares}
\bY_i=u(\alpha_i),
\ee
where 
 \be{alpha}
 \alpha_i=\gamma^{i-1},
 \ee
 is sent to node $i$, for $i \in [N]$. The choice of $\alpha_i$'s and $\beta_j$'s are demonstrated in Figure\,\ref{circle} for the choice of parameters $k=6$, $t=2$, and $N=16$ in the complex plane. It will be clarified in Section\,\ref{Analytical results} that the specific choice of $\beta_j$'s according to \eqref{beta} enables characterizing a closed-form upper bound on the \emph{absolute error} of the outcomes of ALCC. In this context, the absolute error is the magnitude of the difference between the ALCC outcome in a practical setting and the true result of the computation. 
  
  \begin{figure}[t]
  	\begin{center}
  		\includegraphics[width=7cm]{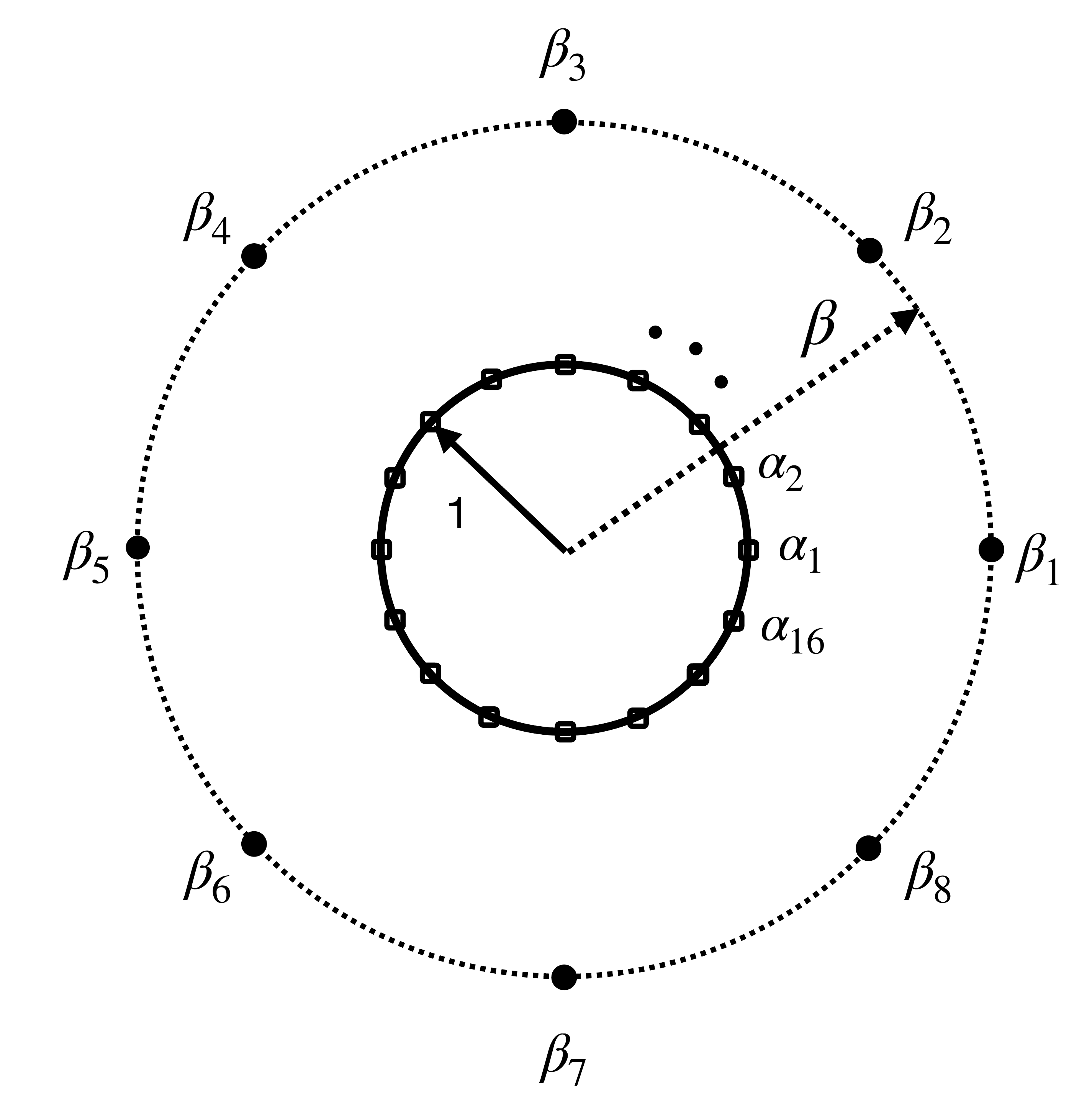} 
  		\caption{\small Demonstration of  the locations of $\alpha_i$'s and $\beta_j$'s in the complex plane for $k=6, t=2, N=16$. Both circles are centered at the origin.}  \label{circle}
  	\end{center}
  \end{figure}
  
Next, we discuss the decoding step during which the master node recovers the desired outcome by collecting and processing the results returned by a sufficient number of worker nodes. The $i$-th node computes $f(\bY_i)$ and returns the result back to the master node. The master node then recovers $f(\bX_i)$, for $i \in [k]$, in two steps. In the first step, it recovers the polynomial $f(u(z))$ by using the results returned from at least $(k+t-1)D+1$ worker nodes. Note that this is the minimum number of returned evaluations needed to guarantee a successful interpolation of $f(u(z))$ since $f(u(z))$ has degree $(k+t-1)D$. For ease of notation, let 
\be{Dtilde-def}
\tilde{D} = (k+t-1)D.
\ee
In the second step, to recover $f(\bX_i)$'s, the master node computes $f(\beta_j)$ for $j \in [k]$. Note that $u(\beta_j)=\bW_j$ for $j \in[k+t]$, since $l_j(\beta_i)$ is $1$ for $i=j$ and is zero otherwise.

The ALCC protocol, as described above, can also take into account the issue regarding stragglers same as how it is done in LCC \cite{yu2019lagrange}. Let the maximum number of stragglers be denoted by $s$. Hence, the number of computational parties is assumed to be $N=\tilde{D}+s+1$. 
 
\noindent
\begin{remark}\label{floating-point}
In theory, if the computations are done over the complex numbers with infinite precision, then $f(\bX_i)$'s are computed accurately. In practice, however, data is represented using a finite number of bits, either as fixed point or floating point. We assume floating-point representation for data symbols and operations involving them in our analysis. This is more suitable to mimic operations over complex (real) numbers. Let $b_m$ denote the number of precision bits  in the floating-point representation, referred to as the  \emph{mantissa}, and $b_e$ denote the number of bits used to represent the exponent part, referred to as the  \emph{exponent}. Note that the entries of the noise matrices $\bN_i$'s are bounded in practice. In other words, for practical purposes, it is assumed that the entries of the noise matrices are drawn from the Gaussian distribution that is truncated to $\left[-\theta \frac{\sigma_n}{\sqrt{t}},\ \theta \frac{\sigma_n}{\sqrt{t}}\right]$, for some $\theta \in \R$.   
 \end{remark}
 %==============================================================================%
 %                                                                              %
 %   1.Accuracy Analysis                                                           %
 %                                                                              %
 %==============================================================================%
 \section{Accuracy Analysis } 
 \label{sec:Accuracy}
In this section, accuracy of the final outcome of ALCC, specified in Section \ref{sec:System Model}, is characterized in terms of various other parameters of the scheme. This is done by assuming that floating-point numbers are used to represent data symbols and to carry out operations involving them.

\subsection{Analytical results}\label{Analytical results}

We start by providing an alternative characterization for the Lagrange monomials, defined in \eqref{Lagrange_monomials}, given the certain values for $\beta_j$'s specified in \eqref{beta}. This is done in the following lemma. 
     \begin{lemma}\label{L_j lemma}
     	For all $j \in [k+t]$, we have
     	\be{Lagrange_monomial_statndard}
     	l_j(z)=\frac{1}{k+t} \sum_{l=0}^{k+t-1}(\frac{z}{\beta_j})^{l}.
     	\ee
     \end{lemma}
     
     \begin{proof}
     	Using \eqref{Lagrange_monomials}, one can write
     	\be{l_j-rewrite}
     	l_j(z)=\prod_{l\in [k+t] \setminus j} \frac{\frac{z}{\beta_j}-\frac{\beta_l}{\beta_j}}{1-\frac{\beta_l}{\beta_j}}=\prod_{h=1}^{k+t-1} \frac{\frac{z}{\beta_j}-\omega^h}{1-\omega^h}.
     	\ee
     Note that $\omega^h$, for $h=0,1,\hdots,k+t-1$, is a $(k+t)$-th root of unity. Hence, we have
     	\be{x_formula}
     	\prod_{h=1}^{k+t-1} (x-\omega^h) = \frac{x^{k+t}-1}{x-1}= \sum_{h=0}^{k+t-1}x^h.
     	\ee
     	Using \eqref{x_formula} one can write 
     	\be{num}
     	\prod_{h=1}^{k+t-1} (\frac{z}{\beta_j})-\omega^h=\frac{(\frac{z}{\beta_j})^{k+t}-1}{\frac{z}{\beta_j}-1}=\sum_{h=0}^{k+t-1}(\frac{z}{\beta_j})^h,
     	\ee
     	and 
     	\be{den}
     	\prod_{h=1}^{k+t-1} (1-\omega^h)=\sum_{h=0}^{k+t-1}1=k+t.
     	\ee
     	Combining \eqref{num} and \eqref{den} completes the proof. 
     \end{proof}
 
     %\Lref{L_j lemma}  of $\beta_j$'s enables us to simplify the characterization of  Lagrange monomials. 
     %Let $\bC=(\bC_0, \hdots, \bC_{k+t-1})$ denote the vector of coefficients of the Lagrange polynomial defined in \eqref{Lagrange_polynomial}, i.e., 
     %\be{Lagrange_coeff}
     %\bU(z)=\sum_{h=0}^{k+t-1}\bC_h z^h.
     %\ee
     
      In the following lemma, we use \Lref{L_j lemma} to characterize the coefficients of Lagrange polynomial in terms of $\bW$ and other parameters of the scheme. The result will be used later to derive an upper bound on the absolute error of the outcome of ALCC.

     \begin{lemma}\label{Lagrange_polynomial_standard}
     	The Lagrange polynomial, as specified in \eqref{Lagrange_polynomial}, can be written as 
     	\be{Lagrange_polynomial_rearranged}
     	u(z)=\sum_{l=0}^{k+t-1} \frac{\ti{\bW}_l}{\beta^l}z^l,
     	\ee
     	where 
     	\be{DFT_W}
     	 \ti{\bW_l}\deff\sum_{j=0}^{k+t-1} \bW_j \omega^{-jl}.
     	\ee
     	  \end{lemma} 
     	\begin{proof}
     		The proof is by combining \eqref{Lagrange_polynomial} and \Lref {L_j lemma} as follows:
     		\begin{align}
     		u(z)&=\sum_{j=0}^{k+t-1} \bW_j \frac{1}{k+t}\sum_{l=0}^{k+t-1}(\frac{z}{\beta_j})^{l}\\
     		&=\frac{1}{k+t}\sum_{j,l=0}^{k+t-1}\bW_j (\frac{z}{\beta})^l \omega^{-jl}\\
     		&=\frac{1}{k+t}\sum_{l=0}^{k+t-1} (\frac{z}{\beta})^l (\sum_{j=0}^{k+t-1} \bW_j \omega^{-jl})\\
     		&=\frac{1}{k+t}\sum_{l=0}^{k+t-1} \frac{\ti{\bW_l}}{\beta^l}z^l.
     		\end{align}
     	\end{proof}
     	
 Let $w_{gl}^j$ and $\ti{w}_{gl}^j$ denote the $(g,l)$ entry of $\bW_j$ and $\ti{\bW}_j$, for $(g,l) \in [m]\times[n]$, respectively. Then \eqref{DFT_W} implies that $(\ti{w}_{gl}^0, \hdots, \ti{w}_{gl}^{k+t-1})$ is the discrete Fourier transform (DFT) of  $(w_{gl}^0, \hdots, w_{gl}^{k+t-1})$. Hence, the encoder can utilize the fast algorithms developed for the DFT implementation to compute $\ti{\bW}_j$'s and then computes the shares sent to the nodes according to  \eqref{Lagrange_polynomial_rearranged}, see, e.g., \cite{nussbaumer1981fast}.
 
 The decoder's task is to interpolate the polynomial $f(u(z))$ followed by evaluating it over $\alpha_i$'s, for $i\in[N]$. Let the polynomial $f(u(z))$ be expressed as 
 \be{f}
 f(u(z))= \sum_{i=0}^{\tilde{D}} \bV_i z^i,
 \ee 
 with $\bV_i \in \R ^{u \times h}$. Let $A=\{i_1, \hdots, i_{\tilde{D}+1}\}$ denote the indices of \textit{non-straggler} users, i.e., users that have returned their computation results to the master node. The interpolation step at the decoder is equivalent to inverting the following matrix:
\be{B-def}
 \bold{B}_{(\tilde{D}+1)\times (\tilde{D}+1)} \ \deff\
 \begin{bmatrix}
 1 & \gamma^{i_1} & \gamma^{2i_1} & \hdots &\gamma^{\tilde{D}i_1}\\
 1 & \gamma^{i_2} & \gamma^{2i_2} & \hdots &\gamma^{\tilde{D}i_2}\\
 \vdots & \vdots & \vdots & \vdots &\vdots\\
 1 & \gamma^{i_{\tilde{D}+1}} & \gamma^{2i_{\tilde{D}+1}} & \hdots &\gamma^{\tilde{D}i_{\tilde{D}+1}}
 \end{bmatrix}.
\ee
 
%\noindent \textbf{Remark 2.} 
\begin{remark}\label{perturbation}
In general, in a system of linear equations $\bA\bx=\boldsymbol{y}$, where $\bx$ is a vector of unknown variables and $\bA$ is referred to as \emph{coefficient matrix}, the perturbation in the solution caused by the perturbation in $\boldsymbol{y}$ is characterized as follows. Let $\hat{\boldsymbol{y}}$ denote a noisy version of $\boldsymbol{y}$, where the noise can be caused by round-off errors, truncation, etc. Let also $\hat{\bx}$ denote the solution to the considered linear system when $\boldsymbol{y}$ is replaced by $\hat{\boldsymbol{y}}$. Let $\Delta \bx\deff \hat{\bx}-\bx$ and $\Delta \boldsymbol{y}\deff \hat{\boldsymbol{y}}-\boldsymbol{y}$ denote the perturbation in $\boldsymbol{x}$ and $\boldsymbol{y}$, respectively. Then the relative perturbations of $\boldsymbol{x}$ is bounded in terms of that of $\boldsymbol{y}$ as follows \cite{demmel1997applied}:
\be{relative_error}
\frac{\norm{\Delta \bx}}{\norm{\bx}}\leq \kappa_{\bA}\frac{\norm{\Delta \boldsymbol{y}}}{\norm{ \boldsymbol{y}}},
\ee
where $\kappa_{\bA}$ is the condition number of $\bA$ and $\norm{.}$ denotes the $l^2$-norm.
\end{remark}
 For $(g,l) \in [u] \times [h]$, let $v_{gl}^i$ denote the $(g,l)$ element of $\bV_i$ for $i = 0,1,\dots,\tilde{D}$, where $\bV_i$ is specified in \eq{f}, and $f_{gl}^j$ denote the $(g,l)$ element of $f(\bX_j)$ for $j \in [k]$. Let also 
 $$
 \bv_{gl} \deff(v_{gl}^0, \hdots, v_{gl}^{\tilde{D}}),
 $$
 for $g \in [u]$ and $l \in [h]$. For ease of notation, let
 \be{beta_def}
 \overline{\beta} = \frac{\beta^{\tilde{D}+2}-1}{\beta^2-1}.
 \ee
 The following lemma establishes a relation between the error in the entries of the outcome of the scheme, i.e., $f_{gl}^j$, and the entries of $\bV_i$, i.e., $v_{gl}^i$. Note that in this analysis the error due to representing $\alpha_j$'s and $\beta_j$'s using floating-point numbers is discarded as it is dominated in practice by the error imposed by the precision loss in the elements of $\bV_i$'s, specified in \eqref{f}.
 
  \begin{lemma}
  \label{lem3}
For all $g\in [u]$, $l \in [h]$, and $j \in [k]$ we have  	
	  	  \be{outcome_error}
	\Delta f_{gl}^j \leq \overline{\beta} \norm{\bv_{gl}} \kappa_{\bB} 2^{-b_m},
	\ee
where $\overline{\beta}$ is defined in \eqref{beta_def}, $\bB$ is defined in \eqref{B-def}, and $b_m$ is the number of precision bits in the floating-point representation, specified in \Rref{floating-point}. %Remark\,1.
  	\end{lemma}
  \begin{proof}
  Let $\tilde{f}_{gl}^i \deff f_{gl}(u(\alpha_i))$ for all $i\in [N]$, $g\in[u]$ and $l\in[h]$, and $\tilde{\bff}_{gl}\deff (\tilde{f}_{gl}^{i_1}, \dots, \tilde{f}_{gl}^{i_{\tilde{D}+1}})$, where $i_1, \dots, i_{\tilde{D}+1}$ represent the indices of worker nodes that returned the computation results. Note that the evaluations of \eqref{f} over $\alpha_i$'s, for $i\in A$, can be regarded as $u\times h$ systems of linear equations all with $\bB$ as the underlying coefficient matrix, i.e., 
  \be{set-of-eqs}
  \tilde{\bff}_{gl}=\bB \bv_{gl},
  \ee
  for all $g\in[u]$ and $l\in[h]$. By utilizing the statement in \Rref{perturbation} one can write
 \be{perturb-result}
\frac{\norm{\Delta \bv_{gl}}}{\norm{\bv_{gl}}}\leq \kappa_{\bB}\frac{\norm{\Delta \tilde{\bff_{gl}}}}{\norm{ \tilde{\bff_{gl}}}}.
\ee
Note that the precision error in the considered floating-point numbers is bounded by $2^{-b_m}$ since it is assumed that no other error is imposed on the computation results in the worker nodes. Hence, one can write 
  \be{assumption}
  \frac{\norm{\Delta \tilde{\bff_{gl}}}}{\norm{ \tilde{\bff_{gl}}}}\leq 2^{-b_m}.
  \ee
   Combining \eqref{perturb-result} with \eqref{assumption} results in
  	\be{elementwise_error}
  	\frac{\norm{\Delta \bv_{gl}}}{\norm{\bv_{gl}}}\leq \kappa_{\bB} 2^{-b_m},
  	\ee
  	for all $g\in[u]$ and $l\in[h]$. Moreover, note that 
  $$
  f(\bX_j)=f(u(\beta_j))= \sum_{i=0}^{\tilde{D}} \bV_i \beta_j^i.
  $$
  Let $\boldsymbol{\beta}_j$ denote $(1,\beta_j, \beta_j^2, \hdots, \beta_j^{\tilde{D}})$, for $j \in [k]$. Then one can write
  	\be{dot}
  	f_{gl}^j=\boldsymbol{\beta_j} \cdot\bv_{gl},
  	\ee
which implies that
  	\be{dot_error}
  	\Delta f_{gl}^j \leq \norm{\boldsymbol{\beta}_j} \norm{\Delta\bv_{gl}},
  	\ee 
where $\cdot$ denotes the inner product operation. Note that for all $j \in [k]$,
  	\be{beta_sum}
  	\norm{\boldsymbol{\beta}_j}^2\leq \sum_{i=0}^{\tilde{D}} \beta^{2i}=\frac{\beta^{\tilde{D}+2}-1}{\beta^2-1} = \overline{\beta}.
  	  	\ee
Combining \eqref{elementwise_error}, \eqref{dot_error} and \eqref{beta_sum} yields
  	  \be{}
  	  	\Delta f_{gl}^j \leq \overline{\beta} \norm{\bv_{gl}} \kappa_{\bB} 2^{-b_m},
  	  	\ee
  	  	which completes the proof.
  \end{proof}

Let $c_{ij}$ denote the maximum absolute value of the coefficients of $f_{ij}(\cdot)$ for all $i \in [u]$ and $j \in [h]$, $c \deff \max_{i,j} c_{ij}$ , and $\lambda_{\min}$ denote the minimum singular value of $\bB$, defined in \eqref{B-def}. In the next theorem, an upper bound on the absolute error in the outcome of the protocol is provided for the general class of polynomials over matrices, defined in \eqref{polynomial-def}. 

\begin{theorem}\label{accuracy}
	The absolute error on the entries of $f(\bX_j)$, for $j \in [k]$, in the outcome of ALCC is bounded as follows:
	\be{accuracy_bound}
	\Delta f_{gl}^j \leq \overline{\beta} \frac{c (mne)^D}{\lambda_{\min}} \sqrt{\tilde{D}+1} (kr+t\theta \sigma_n)^D\kappa_{\bB} 2^{-b_m}(1+O(\frac{1}{\sigma_n})),
	\ee
where $\overline{\beta}$ is defined in \eq{beta_def},  $\bB$ is defined in \eqref{B-def}, and $b_m$ is the number of precision bits in the floating-point representation, specified in \Rref{floating-point}.%Remark\,1.	
\end{theorem}
\begin{proof}
Note that \eqref{set-of-eqs} implies 
\be{svd-bound}
\norm{\bv_{gl}}\leq\frac{\norm{\tilde{\bff}_{gl}}}{\lambda_{\min}},
\ee
for any arbitrary set of non-straggler indices $A$. Also, \Lref{Lagrange_polynomial_standard} implies that the $j$-th entry of the DFT of $(w_{gl}^0, \hdots, w_{gl}^{k+t-1})$ is equal to $\ti{w}_{gl}^j$. Hence, we have
 \be{DFT_bound}
 \ti{w}_{gl}^j \leq kr+t\theta \sigma_n,
 \ee 
 where we used the fact that the absolute value of entries of $\bX_j$'s and $\bN_j$'s are less than $r$ and $\theta \frac{\sigma_n}{\sqrt{t}}$, respectively, as discussed in Section\,\ref{sec:System Model}. Moreover, for $i \in [N]$ one can write 
\be{c-bound}
\tilde{f}_{gl}^i 
\leq c (mne)^D (kr+t\theta \sigma_n)^D(1+O(\frac{1}{\sigma_n})),
\ee
which holds by noting that the number of total monomials of degree $D$ in $mn$ variables is equal to ${{mn+D-1}\choose{D-1}}$ and we have
$$
{{mn+D-1}\choose{D-1}}\leq (\frac{e(mn+D-1)}{D})^D\leq (emn)^D,
$$
where $e$ is the natural number. Then \eqref{c-bound} implies that
\be{bff-bound}
\norm{\tilde{\bff}_{gl}}\leq c (mne)^D \sqrt{\tilde{D}+1} (kr+t\theta \sigma_n)^D(1+O(\frac{1}{\sigma_n})),
\ee
since $\tilde{\bff}_{gl}$ has $\tilde{D}+1$ components. Substituting \eqref{bff-bound} into \eqref{svd-bound} together with the result of \Lref{lem3} complete the proof.  
\end{proof}

\Tref{accuracy} provides an upper bound on the accuracy of the outcome of ALCC with floating-point implementation for a general polynomial function $f(\cdot)$. However, the polynomial $f(\cdot)$ often has a certain structure in practice that can be leveraged to strengthen the result of \Tref{accuracy}. More specifically, we say that $f(\cdot)$ is a \emph{matrix polynomial function}, or simply a \textit{matrix polynomial}, if it can be expressed by matrix addition, multiplication, and transposition as well as addition and multiplication by a constant matrix/vector/scalar. For instance, $f(\bX)=\ba\bX\bX^t$, for some vector $\ba$, is such a matrix polynomial function. The difference between a general polynomial, defined in \eqref{polynomial-def}, and a matrix polynomial is illustrated in the following example. 
\begin{exmp}
Let $\bX= \begin{bmatrix}
x_{11}& x_{12}\\
x_{21}& x_{22}
\end{bmatrix}$. Then the function $g_1(\bX)\deff\begin{bmatrix}
x_{11}^2+x_{11}x_{12}+x_3^2 & x_1\\
x_{21}^2+x_{11}x_{22} & x_{22}^2
\end{bmatrix}$ is a polynomial function of degree $2$, as defined in \eqref{polynomial-def}, but is not a matrix polynomial function. The function $g_2(\bX)\deff\begin{bmatrix}
x_{11}^2+x_{12}^2 & x_{11}x_{21}+x_{12}x_{22}\\
x_{21}x_{11}+x_{22}x_{12} &x_{21}^2+ x_{22}^2
\end{bmatrix}=\bX\bX^t$ is a matrix polynomial function of degree $2$. Note also that the determinant of a matrix, i.e., $g_3(\bX)\deff\textit{det}(\bX)=x_{11}x_{22}-x_{12}x_{21}$ is a polynomial but not a matrix polynomial. The matrix inversion function, i.e., $g_4(\bX)\deff\bX^{-1}$ is not even a polynomial function.
\end{exmp}
   The following corollary provides  a stronger accuracy bound on the outcome of ALCC with matrix polynomial as its underlying function.
   \begin{corollary}\label{corollary}
   Let $f(\cdot)$ be a matrix polynomial function. Then, the absolute error on the entries of $f(\bX_j)$, for $j \in [k]$, in the outcome of ALCC is bounded as follows:
    \be{accuracy_bound2}
	\Delta f_{gl}^j \leq C (kr+t\theta \sigma_n)^D\kappa_{\bB} 2^{-b_m}(1+O(\frac{1}{\sigma_n})),
	\ee
	where $C\deff\overline{\beta} \frac{c \max(m,n)^D}{\lambda_{\min}} \sqrt{\tilde{D}+1} $.
   \end{corollary}
   \begin{proof}
       Note that the number of $D$-degree monomials in a matrix polynomial is at most $\max(m,n)^D$. The remaining of the proof is similar to that of \Tref{accuracy}.
   \end{proof} 

\begin{remark}\label{straggler}
Note that when no stragglers are assumed, i.e., $s=0$, picking  $\alpha_i$'s in the proposed ALCC protocol according to \eqref{alpha} implies that the matrix $\bB$, defined in \eqref{B-def}, is a unitary matrix. Hence, we have $\kappa_{\bB}=1$ which is the minimum possible for the condition number $\kappa_{\bB}$. For the case of ALCC with stragglers, i.e., $s>0$, one can utilize the following upper bond on $\kappa_{\bB}$ \cite[Theorem 1]{ramamoorthy2019numerically}:
\be{kappaB_bound}
\kappa_{\bB} \leq O(\tilde{N}^{s+6}),
\ee
where $\tilde{N}$ is the smallest odd number larger than $N$. Combining \eqref{kappaB_bound} with \eqref{accuracy_bound} leads to an upper bound on the accuracy of ALCC scheme with $s$ stragglers. %In particular, it can be observed that for a fixed number of stragglers $s$, the number of true digits in the computation outcome of ALCC decreases with rate at most $O(\log(N))$ in the number of worker nodes $N$. 
\end{remark}

\subsection{Comparisons with LCC and numerical results}
\label{three:B}

In this section we compare the accuracy of ALCC with that of LCC that employs finite field operations. In LCC the computations are preformed over a finite field of a prime size $p$, denoted by $\F_p$, and are implemented using fixed-point numbers \cite{yu2019lagrange}. In order to have a fair comparison, we assume that the number of bits that are used to represent data symbols in ALCC, that uses floating point, and LCC, that uses fixed point, are equal. Let that number be denoted by $b$. It is also assumed that $b$ is fixed throughout the implementation of the scheme. It is shown in Section\,\ref{Analytical results} how the accuracy of ALCC depends on $b$. Same as in ALCC, the accuracy of LCC also depends on $b$ as discussed next. 

In LCC the real-valued data are assumed to be first quantized and then mapped to elements of $\F_p$. Also, note that if a symbol computed during the process by one of the workers in LCC becomes larger than $p$, an incident referred to as an overflow error, then a successful recovery of the outcome of the computation can not be guaranteed. Let $\Delta$ denote the corresponding quantization step. Let also $s_{ij}$ denote the sum of the absolute values of the coefficients of the polynomial $f_{ij}$, for $i \in [u]$ and $j \in [h]$, and let $s_a\deff\max_{i,j} s_{ij}$. Then, in order to avoid overflow errors, it is required that
\be{wrap_around_criterion}
\frac{s_a}{\Delta}(\frac{r}{\Delta})^D\leq \frac{p}{2}, 
\ee
since the left hand-side of \eqref{wrap_around_criterion} corresponds to the maximum value of the polynomial $f(\cdot)$ when evaluated over the quantized data. The latter is by noting that $f_{ij}(x)\leq s_ax^D\leq s_a r^D$ for all $i,j$, the quantization step is $\Delta$, and the magnitude of the entries of $\bX$ is upper bounded by $r$. 

The next step is to characterize how large $p$ can be given the fixed number of representation bits $b$. To this end, two different scenarios can be considered regarding how the intermediate multiplications, at the worker nodes, are carried out. More specifically, the intermediate multiplications may or may not be done modulo $p$. We consider the two cases separately and provide bounds on the accuracy in both cases. Performing intermediate multiplications modulo $p$ leads, in general, to a better accuracy, as will be also shown in the remaining of this section. However, this improvement comes at the cost of increased \emph{latency} of the fixed-point implementation. This is because, in practice, performing multiplications over large finite fields require further processing, compared to the regular multiplication, and are slower than the regular multiplications.
%This exhibits a trade-off between the \emph{accuracy} and \emph{latency} of the fixed-point implementation where these two scenarios can be viewed as the extreme points of this trade-off. Employing modular multiplication each time two symbols are multiplied during the protocol provides the best accuracy bound while compromising the speed of LCC.

In the first case, it is assume that the intermediate multiplications are done modulo $p$. Then in order to avoid overflow errors in multiplications while employing fixed-point implementation, it is necessary to have
\be{overflow_modular}
 p^2\leq 2^b.
\ee

In the second case, it is assumed that the underlying multiplications are done over $\Z$ and the worker nodes need to compute the result modulo $p$ only once after the polynomial evaluations are completed. In this case, the condition in \eqref{overflow_modular} is modified as follows:
\be{overflow_regular}
\frac{s_a}{\Delta}p^D\leq 2^b.
\ee
%which should be satisfied in order to ensure that overflow errors do not occur in the intermediate steps. 

Combining \eqref{wrap_around_criterion} with \eqref{overflow_modular}, for the first case, and with  \eqref{overflow_regular}, for the second case, provides lower bounds on the absolute error of the outcomes of LCC. In particular, one must have 
 \be{modular_accuracy}
 (\frac{s_ar^D}{2^{(\frac{b}{2}-1)}})^{\frac{1}{D+1}}\leq \Delta,
 \ee
 for the first case, and
 \be{regular_accuracy}
 (\frac{s_a^{(1+\frac{1}{D})}r^D}{2^{(\frac{b}{D}-1)}})^{(\frac{D}{D^2+D+1})} \leq \Delta,
 \ee
 for the second case. 
 
Now, consider ALCC with floating-point implementation where each symbol is represented by $b$ bits. In current standard systems, $8$ bits are allocated to represent the exponent. Also, one bit is reserved for indication of zero and one bit is reserved for the sign flag. Hence, the total number of precision bits $b_m$ is equal to $b-10$. We use these parameters to plot the bounds on the accuracy of ALCC versus that of LCC. In Figure\,\ref{FLPvsFXP_acc}, the upper bound on the absolute error in ALCC with floating-point implementation, provided in \Cref{corollary}, is plotted and is compared with the lower bounds on the absolute error in LCC with fixed-point implementation, provided in \eqref{modular_accuracy} and \eqref{regular_accuracy}, for the two aforementioned cases. Note that for the experiments with the results shown in Figure\,\ref{FLPvsFXP_acc} the terms $O(\frac{1}{\sigma_n})$, that are used in the bounds provided in \Tref{accuracy} and \Cref{corollary}, are equal to zero due to the certain matrix polynomial function considered. In general, such terms can be often discarded in ALCC with general underlying polynomial functions as $\sigma_n$ considered in practice is relatively large due to privacy concerns, e.g., $\sigma_n=10^{12}$ in the considered experiments.  Note that these bounds are plotted as a function of $b$, i.e., the total number of bits reserved to represent a data symbol in both the fixed-point and the floating-point implementation while the other parameters of the system are fixed. It can be observed that for $b$ larger than a certain threshold, the upper bound derived on the error in ALCC is smaller than both of the lower bounds on the error in LCC. Since these bounds may not be tight, the actual threshold would perhaps be lower than what is shown in Figure\,\ref{FLPvsFXP_acc}.

%Note that one can utilize a different strategy on \emph{how often} the modulo operation is performed during the implementation of the scheme thereby providing schemes attaining the points between the extreme points studied in this paper. Hence, the corresponding absolute error plot for such schemes lie between the two plots in Figure\,\ref{FLPvsFXP_acc} associated with the extreme cases. Note that in Figure\,\ref{FLPvsFXP_acc}, the \emph{upper} bound of error in floating-point implementation is compared with the \emph{lower} bound on the error in fixed-point implementation. Hence, this comparison does not necessarily imply that the actual error in ALCC is worse than that of LCC in the region that it's corresponding plot is above the other plots. The actual errors of ALCC and LCC are compared in Section\,\ref{label} for a certain computational task.

 One can also analyze the aforementioned bounds in terms of the decay rate of the absolute error as $b$ increases. More specifically, the lower bounds on the absolute error in LCC, derived in \eqref{overflow_regular} and \eqref{modular_accuracy}, decay exponentially in $b$ with an exponent between $\frac{1}{2(D+1)}$ and $\frac{1}{D^2+D+1}$. However, the upper bound on the absolute error in ALCC, derived in \eqref{accuracy_bound}, decays exponentially in $b$ with exponent $1$. This significant improvement in accuracy comes at the expense of deviating from the perfect privacy, in an information-theoretic sense, in ALCC compared to LCC. This will be discussed in details in the next section. In particular, it will be shown that this deviation is negligible for practical purposes.

 \begin{figure}[t]
 	\vspace{1.5mm}
 	\begin{center}
 		\includegraphics[width=\linewidth]{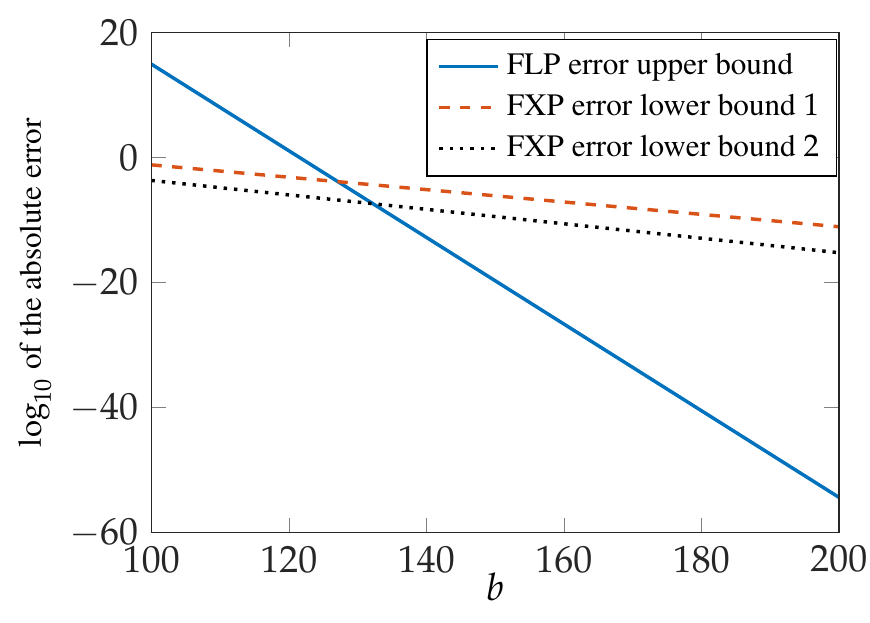}
 		\vspace{-2mm}
 		\caption{\small Comparison of upper bound on the absolute error in floating-point (FLP) implementation with lower bounds on the error in fixed-point (FXP) implementation employing conventional (FXP $1$) and (FXP $2$). The underlying function considered is $f(\bX)=\bX\bX^t$. The parameters are as follows: $ k=5, t=3, s=0,  m=n=1000, r=100, \theta=3, \sigma_n=10^{12}$. }\label{FLPvsFXP_acc}
 	\end{center}
 	\vspace{-7mm}
 \end{figure}

 %==============================================================================%
%                                                             %
%   1.Privacy Analysis                                                           %
%                                                                              %
%==============================================================================%

\section{Privacy Analysis } \label{Sec:privacy}

In this section, we analyze the privacy level of data in ALCC by considering two notions of security, namely, the mutual information security (MIS) and the distinguishing security (DS) over the continuous probability space. 

\subsection{Privacy analysis with Gaussian noise}\label{privacy_gaussian}

We first characterize the privacy of ALCC in terms of the MIS metric by utilizing existing results on the capacity of multiple-input-multiple-output (MIMO) channel. Furthermore, by using the relation between the MIS and the DS security metrics in the analog domain the privacy of ALCC is also characterized in terms of the DS metric. Such a relation is observed in the context of wiretap channels in \cite{ling2014semantically} and has been also utilized in \cite{soleymani2020privacypreserving}. 

Consider the ALCC protocol described in Section\,\ref{sec:System Model}. For $j \in [k]$ and $i \in [t]$, let $X_j$ and $N_i$ denote the $(g,l)$ element of the matrices $\bX_j$ and $\bN_i$, respectively, for some fixed $g \in [m]$ and $l \in [n]$. For the sake of clarity, $g$ and $l$ are fixed throughout this section. However, the analysis does not depend on the specific choice of $g$ and $l$.

The Lagrange polynomial introduced in \eqref{Lagrange_polynomial} is written for the fixed considered indices as follows:
\be{comp-wise Lagrange}
U(z)=\sum_{j=1}^k X_j l_j(z)+\sum_{j=k+1}^{k+t} N_{j-k} l_j(z),
\ee
where data symbols $X_j$'s are random variables with arbitrary distribution and with the range $\Bbb{D}_x\deff [-r,r]$, for $j \in [k]$, and $N_i \sim \mathcal{N}(0,\frac{\sigma_n^2}{t})$, for $i\in [t]$. Let $Y_i$ denote the corresponding entry of $\bY_i$, for $i \in [N]$. In other words, $Y_i=U(\alpha_i)$. Let $T = \{ i_1, \cdots, i_t\}$ denote the set of indices for the colluding parties. Let also $X$, $N$, and $Y_T$ denote $(X_1, \cdots, X_k)^{\text{T}}$, $(N_1, \cdots, N_t)^{\text{T}}$, and $(Y_{i_1}, \cdots, Y_{i_t})^{\text{T}}$, respectively, where $(\cdot)^{\text{T}}$ is the transpose operation. By convention, a random variable/vector is denoted by a capital letter and its instance is denoted by the corresponding lower case letter. 

The following equation relates the encoded symbols received by the colluding set of parties $T$ to the dataset symbols and the added noise symbols:
\be{matrix_form}
Y_T=\bL_T X+ \tilde{\bL}_T N,
\ee
where 
 \be{LAdeff}
\bL_T\ \deff\
\begin{bmatrix}
l_1(\alpha_{i_1}) & \hdots &l_k(\alpha_{i_1})\\
l_1(\alpha_{i_2}) & \hdots &l_k(\alpha_{i_2})\\
\vdots & \vdots &\vdots\\
l_1(\alpha_{i_t}) & \hdots &l_k(\alpha_{i_t})\\
\end{bmatrix}_{t\times k} ,
\ee
and
\be{LAtdeff}
\tL_T\ \deff\
\begin{bmatrix}
l_{k+1}(\alpha_{i_1}) & \hdots &l_{k+1}(\alpha_{i_1})\\
l_{k+2}(\alpha_{i_2}) & \hdots &l_{k+2}(\alpha_{i_2})\\
\vdots & \vdots &\vdots\\
l_{k+t}(\alpha_{i_t}) & \hdots &l_{k+t}(\alpha_{i_t})\\
\end{bmatrix}_{t\times t} .
\ee
The amount of information revealed to the set of colluding parties can be measured in terms of the MIS metric, denoted by $\eta_c$, defined as follows:
\be{MIS_def}
\eta_c \deff \max_{\substack{T}} \max_{\substack{P_X:|X_j|<r, \forall j \in [k]}} I(Y_T; X),
\ee
where $P_X$ is the probability density function (PDF) of $X$ and the maximization is taken over all $T \subset [N]$ with $|T|=t$. Since $|X_j|\leq r$, we have $E[X_j]^2\leq r^2$. Then, one can write
\be{etac_upperbound}
\eta_c \leq \max_{\substack{T}}\max_{\substack{P_X: E[X_j^2]\leq r^2}} I(X;Y_T).
\ee

Next, we characterize the right hand side of \eqref{etac_upperbound} in terms of other parameters of the system. To this end, the capacity results of MIMO channels are utilized as discussed next. 
Consider a MIMO channel with $k$ transmit and $t$ receive antennas and the input-output relation 
\be{MIMO_model}
\by=\bH\bx+\bn,
\ee
where $\bx$ and $\by$ are the $k\times 1$ transmitted signal and the $t \times 1$ received signal vectors, respectively, $\bH_{t \times k}$ represent the channel gain matrix known to both the transmitter and the receiver, and $\bn_{t\times 1}$ is an additive zero-mean Gaussian noise vector. Let $\bN_c$ denote the noise correlation matrix, i.e., the covariance matrix of the vector $\bn$. By using the results on the capacity of MIMO channel with equal-power allocation constraint and correlated noise, one can get an upper bound on the right-hand side of \eqref{etac_upperbound}. The capacity of this MIMO channel, under equal-power allocation constraint, is well-known and is expressed as follows \cite[IV-A]{schumacher2002antenna}:
\be{MIMO-capacity}
	C=\log_2|\bI_t+P\bN_c^{-1}\bH\bH^H|,
\ee
where $P$ is the maximum transmission power of each antenna at the transmitter side, $\bI_t$ is the $t \times t$ identity matrix and  $|\cdot|$ denotes matrix determinant.

\begin{theorem}
\label{mic_bound_thm}
In the proposed ALCC, the MIS metric $\eta_c$, defined in \eq{MIS_def}, is upper bounded as follows:
\be{etac_MIMO}
	\eta_c \leq \max_{\substack{T}} \log_2|\bI_t+\frac{r^2t}{\sigma_n^2}\tilde{\bSig}^{-1}_T\bSig_T|,
	\ee
	where $\tilde{\bSig}_T\deff \tL_T\tL_T^H$ and $\Sigma_T\deff\bL_T\bL_T^H$. Also, $\tilde{\bSig}_T$ and $\tL_T$ are specified in \eq{LAdeff} and \eq{LAtdeff}, respectively. 
\end{theorem}
\begin{proof}
Note that in \eqref{matrix_form} the term $\tL_T N$ can be regarded as the noise vector and, consequently, \eqref{matrix_form} can be turned into an equation similar to \eq{MIMO_model} describing the MIMO channel model. Hence, by using this observation together with \eqref{MIMO-capacity} and the definition of capacity, \eqref{etac_upperbound} leads to \eq{etac_MIMO}.
\end{proof}

\begin{corollary}
For $r=o(\sigma_n)$, we have
	\be{etac_MIMO_approx}
	\eta_c \leq \frac{1}{\ln(2)}\max_{\substack{T}}\text{tr}(\tilde{\bSig}^{-1}_T\bSig_T)\frac{r^2t}{\sigma_n^2}+o(\frac{r^2}{\sigma_n^2}).
	\ee 
\end{corollary}
\begin{proof}
The proof is by utilizing 
$$
|\bI_t+\epsilon \bA| =1+ \epsilon \text{tr}(\bA)+ o(\epsilon)
$$
together with
$$
\log_2 (1+\epsilon)=\frac{\epsilon}{\ln(2)}+o(\epsilon)
$$
in the upper bound presented in \Tref{mic_bound_thm}. 
\end{proof}

Next, we characterize the privacy of ALCC in terms of the DS metric. The DS metric is defined using the notion of the \emph{total variation} (TV) distance $D_{TV}(.,.)$. In general, for any two probability measures $P_1$ and $P_2$ on a $\sigma$-algebra $\cF$, the TV distance is defined as $D_{TV}(P_1,P_2)\deff \sup_{B \in \cF} |P_1(B)-P_2(B)|$. While DS metric is often defined for discrete random variables in the cryptography literature, it can be also extended to the case of continuous random variables \cite{soleymani2020privacypreserving}. In particular, in the proposed ALCC protocol $\eta_s$ is defined as as follows:
\be{DS_security}
\eta_s\deff \max_{\substack{T}} \max_{\bx_1,\bx_2 \in \Bbb D_X} D_{\text{TV}}(P_{Y_T|X=\bx_1},P_{Y_T|X=\bx_2}),
\ee
where $\Bbb D_X = [-r,r]^k$ is the support of $X$. Note that, roughly speaking, a smaller value for $\eta_s$ implies data is kept more private against any set of $t$ colluding parties. 

Next, we discuss the privacy guarantee for ALCC in terms of the DS metric $\eta_s$. This is done by utilizing relations between $\eta_c$ and $\eta_s$ and the upper bound on $\eta_c$ derived in \Tref{mic_bound_thm}. 

Relations between MIS and DS metrics was first established in \cite{bellare2012cryptographic} though for discrete random variables. In particular, it is shown in \cite{bellare2012cryptographic} that:
\be{eta_relation}
\eta_s \leq \sqrt{2\eta_c},
\ee
 assuming all underlying random variables are discrete. This result is also extended to the analog domain in \cite{ling2014semantically}. In other words, it is shown that \eqref{eta_relation} also holds when the underlying random variables are continuous. Then, combining \eqref{etac_MIMO} with \eqref{eta_relation} yields the following upper bound on the DS metric $\eta_s$:
 \be{eta_s}
 \eta_s \leq \sqrt{2 \max_{\substack{T}} \log_2|\bI_t+\frac{r^2t}{\sigma_n^2}\tilde{\bSig}^{-1}_T\bSig_T|}.
 \ee
 In particular, for $r=o(\sigma_n)$, we have
\be{eta_s approx}
\eta_s \leq \sqrt{\frac{2t}{\ln(2)}\max_{\substack{T}}\text{tr}(\tilde{\bSig}^{-1}_T\bSig_T)} \frac{r}{\sigma_n}+o(\frac{r}{\sigma_n}).
\ee

 \begin{figure*}[h!]
 	\begin{minipage}{0.49\textwidth}
 		\centering
 		\includegraphics[width=\linewidth]{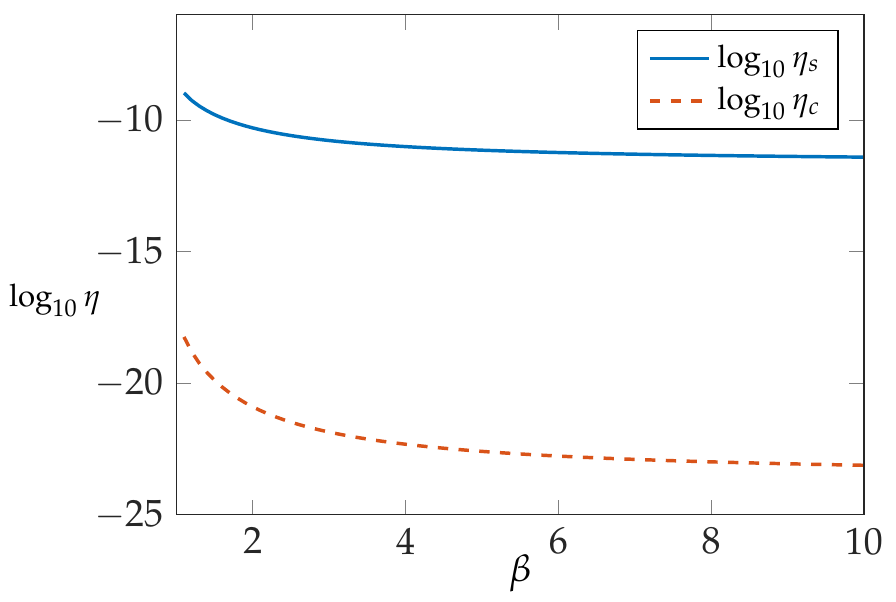}
 		\caption{\small Upper bounds on $\eta_s$ and $\eta_c$ for $N=15, k=4, t=4, \sigma_n =10^{23}, r=10^{10}$.  }\label{etas}
 	\end{minipage}\hfill
 	\begin{minipage}{0.45\textwidth}
 		\vspace{1.5mm}
 		\centering
 		\includegraphics[width=\linewidth]{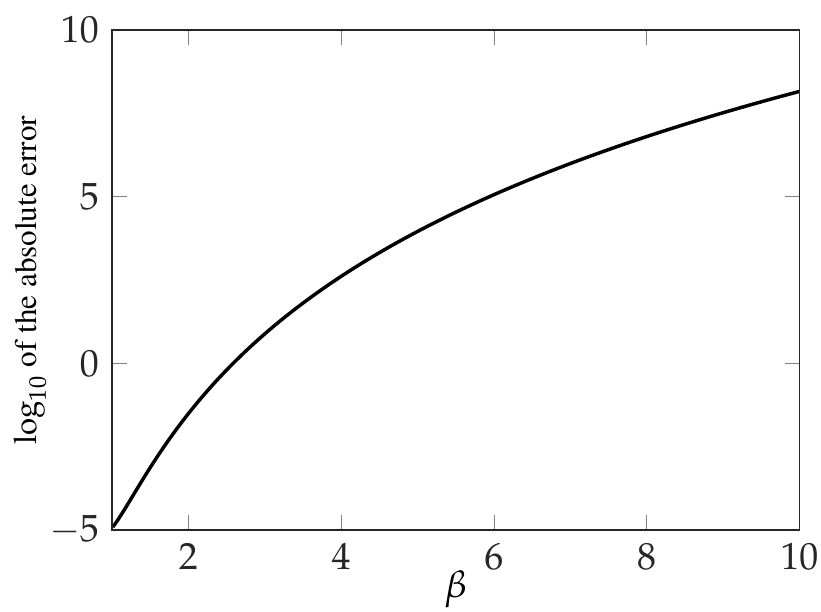}
 		\caption{\small Upper bound on the accuracy of ALCC versus $\beta$ for $D=2, k=4, t=4, s=0, c=1, m=n=1000, r=10^{10}, \theta=3, \sigma_n=10^{23}$ and $v=200$. }\label{FLP_acc}
 	\end{minipage}
 \end{figure*}

%\noindent \textbf{Remark 3.} 
 \begin{remark}\label{acc-pri-trade-off}
Note that both \eqref{etac_MIMO} and \eqref{eta_s} imply that increasing the standard deviation of the added noise, i.e., $\sigma_n$, while other parameters of ALCC are fixed, improves bounds on the privacy level of ALCC. However, this improvement comes at the expense of degrading the accuracy of the outcome of ALCC, according to \Tref{accuracy}. This exhibits a fundamental trade-off between the accuracy and privacy of ALCC. Such a trade-off between accuracy and privacy in the analog domain has been observed for the first time in \cite{soleymani2020privacypreserving} for a privacy-preserving distributed computing setup. 
\end{remark}

Next, the provided upper bounds on the maximum amount of information revealed about the dataset to a subset $T$ of colluding parties with size $t$ are numerically evaluated. This is done for both the MIS security metric, bounded in \eqref{etac_MIMO}, as well as the DS metric, bound in \eqref{eta_s}, for a certain set of parameters and the results are shown in Figure \ref{etas}.  Both $\eta_s$ and $\eta_c$ are plotted versus $\beta$. It can be observed that both the bounds are decreased by increasing $\beta$. In other words, this indicates that increasing $\beta$, in general, leads to enhancement in the privacy of the ALCC protocol. However, the provided upper bound on the accuracy of the outcome of ALCC, provided in \eqref{accuracy_bound} and \eqref{accuracy_bound2}, implies that the precision loss would also grow by increasing $\beta$. The upper bound on the absolute error in ALCC with general underlying matrix polynomial function is plotted versus $\beta$ in Figure\,\ref{FLP_acc} for a certain set of parameters.  Note that, same as in  the experiments with results demonstrated in Figure\,\ref{FLPvsFXP_acc}, the terms $o(\frac{1}{\sigma_n})$ are discarded in the plot in Figure\,\ref{FLP_acc} as well since $\frac{1}{\sigma_n}=10^{-23}$ is negligible.  This together with the plot in Figure\,\ref{etas} demonstrates a new fundamental trade-off between the accuracy and the privacy of the ALCC protocol which is specific to ALCC and is controlled by the choice of $\beta$. It can be also observed from Figure\,\ref{etas} and Figure\,\ref{FLP_acc} that a reasonable value for $\beta$, e.g., $\beta = 1.5$, can be picked for which the upper bounds on $\eta_s$ and $\eta_c$ are reasonably low, e.g., $\sim 10^{-10}$ and $\sim 10^{-20}$, respectively, while the upper bound on the error in the outcome is reasonable for practical purposes, e.g., $\sim 10^{-3}$.

\subsection{Privacy analysis with truncated noise}\label{truncated_noise}
The results presented in Section\,\ref{privacy_gaussian} are derived assuming that the entries of the noise matrices $\bN_i$'s in \eqref {Lagrange_polynomial} are drawn from a complex circular-symmetric Gaussian distribution. However, as discussed in Section\,\ref{sec:Accuracy}, these noise terms should be bounded for practical implementations. In other words, the actual PDF of the noise terms is a properly scaled version of the Gaussian PDF truncated between $-\theta \frac{\sigma_n}{\sqrt{t}}$ and $\theta \frac{\sigma_n}{\sqrt{t}}$, for some $\theta \in \R$. Roughly speaking, we say that the noise terms are truncated. In this section, we extend the results on bounding $\eta_s$ to the case with the noise terms being truncated. We use the upper bound on the DS metric in a similar setup with truncated complex Gaussian noise \cite{soleymani2020privacypreserving} that involves the following quantity:
	\be{mean_dif}
d_{\text{mean}}\,\deff\,\max_{\substack{T}} \max_{\bx_1,\bx_2 \in \Bbb D_X}|\bL_T(\bx_1-\bx_2)|.
\ee 
Note that conditioned on $X=\bx$ in \eqref{matrix_form}, $Y_T$ is a complex Gaussian vector with mean $\bL_T \bx$. Then the parameter $d_{\text{mean}}$, defined in \eqref{mean_dif}, is the maximum Euclidean distance between the means of any two such conditional random vectors in the $t$-dimensional complex vector space, where the maximum is taken over the set of all colluding sets $T$ with size $t$ and all $\bx_1, \bx_2$ in the range of the random vector $X$. In the following lemma an upper bound on $d_{\text{mean}}$ is obtained by using the alternative characterization of Lagrange monomials derived in \Lref{L_j lemma}.

\begin{lemma}\label{d-upperbound}
	The parameter $d_{\text{mean}}$, defined in \eqref{mean_dif}, is upper bounded as follows:
	\be{d-bound}
	d_{\text{mean}}\leq \frac{kr}{k+t}\frac{(\frac{1}{\beta})^{k+t}-1}{(\frac{1}{\beta}) -1}.
	\ee
\end{lemma}
\begin{proof}
	For all $\bx_1, \bx_2 \in \Bbb D_X$ and $T \subset [N]$ with $|T|=t$, we have
	
	\begin{align}
		|\bL_T(\bx_1-\bx_2)|&\leq |\bL_T\bx_1|+|\bL_T\bx_1|\\
		&\leq 2\max_{\bx \in \Bbb D_X}|\bL_T\bx|\label{one}\\
		&\leq 2\sqrt{t}\max_{\substack{\alpha_i:\\i \in [N]}}\max_{ \substack{x_1, \cdots, x_k:\\ |x_i|\leq r} } \sum_{j=1}^k x_j l_j(\alpha_i)\label{two} \\
		&\leq kr \max_{\substack{\alpha_i:\\i \in [N]}} |l_j(\alpha_i) |\label{three}\\
		&\leq \frac{kr}{k+t}\frac{(\frac{1}{\beta})^{k+t}-1}{(\frac{1}{\beta}) -1},\label{four}
	\end{align}
	where \eqref{two} is by the definition of $\bL_T$ in \eqref{LAdeff} and noting that $\bL_T\bx$ is a $t$-dimensional vector, \eqref{three} holds by noting that $|x_i|<r$ and the summation has $k$ terms, and \eqref{four} is by $|x_i|<r$, upper bounding $|l_j(\alpha_i) |$ by \eqref{Lagrange_monomial_statndard} and noting that $|\alpha_i|=1$ for all $i$.
\end{proof}

Let $\eta_s'$ denote the DS metric for the case where the noise terms in \eqref{Lagrange_polynomial} are truncated. The following theorem provides an upper bound on $\eta_s'$ in terms of $\eta_s$, the upper bound on $d_{\text{mean}}$, and other parameters of the ALCC protocol.
\begin{theorem}\label{truncation_bound}
	The DS metric, defined in \eqref{DS_security}, for the case where the entries of $\bN_i$'s in \eqref{Lagrange_polynomial} are drawn from a truncated complex Gaussian distribution with truncation level $\theta \frac{\sigma_n}{\sqrt{t}}$ satisfies the following inequality:
	$$
	\eta'_s \leq \frac{1}{w} \eta_s+\frac{1}{w}(2\exp(-\frac{1}{2}(\theta-\frac{\overline{d}_{\text{mean}}\sqrt{t}}{\sigma_n})^2))^t\label{gaussian_tail_bound_2},
	$$
	where $w= (1-2\exp(-\frac{\theta^2}{2}))^t$ and 	
	$$
	\overline{d}_{\text{mean}}\,\deff\, \frac{kr}{k+t}\frac{(\frac{1}{\beta})^{k+t}-1}{(\frac{1}{\beta}) -1}.
	$$
\end{theorem}
\begin{proof}
    The proof follows by combining \cite[Theorem\,5]{soleymani2020privacypreserving} and \Lref{d-upperbound}.
\end{proof}

A numerical evaluation of the bound provided in \Tref{truncation_bound} implies that, for instance, having $\theta=10$ with $t=10$, together with a very small $\frac{r}{\sigma_n}$, which is often the case in practice, we get $\eta_s' \approx \eta_s $. In other words, the privacy of dataset is not compromised by truncating the noise terms as long as $\theta$ is large enough, e.g., $\theta=10$.

\section{Experiments}\label{sec:experiments}

In this section, we demonstrate the performance of ALCC when applied to a certain computational task through experiments. In the first part of this section, it is shown that the precision of ALCC outcome closely follows that of a \textit{centralized} computation, that is when the computations are done directly at a central node without any encoding and decoding. In particular, it is shown that the accuracy of ALCC is \emph{scalable} with dataset size, i.e., the precision of the results remains \emph{almost} the same for a wide range of sizes of the dataset. In the second part, the performance of LCC \cite{yu2019lagrange} employing fixed-point representation applied to the same computational task is demonstrated. It is shown that the error in the outcome of LCC experiences a sharp increase due to overflow errors as the dataset size passes a certain threshold.

We consider the task of performing a certain matrix-matrix multiplication. For the sake of clarity, we consider computing $\bX^{\text{T}}\bX$ where $\bX \in \R^{m'\times n}$ is a \emph{tall} real-valued matrix, i.e., $m'\gg n$. Such computation is one of the main building blocks in various learning algorithms including training a linear regression model \cite{yu2019lagrange}, or a logistic regression model \cite{soleymani2020privacypreserving, so2019codedprivateml}, etc., where $\bX$ represents a dataset consisting of $m'$ samples in an $n$-dimensional feature space. The matrix $\bX$ can be represented as a batch of matrices $\bX=(\bX_1^{\text{T}}, \cdots, \bX_k^{\text{T}})^{\text{T}}$, where $\bX_i \in \R^{m \times n}$ with $m'=k\times m$. Then we have
$$
\bX^{\text{T}}\bX=\sum_{i=1}^{k} \bX_i^{\text{T}}\bX_i.
$$
Hence, the task of computing $\bX^{\text{T}}\bX$ is reduced to evaluating a degree-$2$ polynomial over a batch of matrices, consisting of $\bX_1$,\dots,$\bX_k$, for which ALCC can be utilized to provide speed up by leveraging the computational power of distributed nodes in parallel. 

Let $\bY$ denote the result of computing $\bX^{\text{T}}\bX$ in a centralized fashion employing floating-point operations. Let also $\bY'$ denote the result of a distributed computing protocol, e.g., ALCC. In order to measure the accuracy loss of the outcome in the distributed protocol compared to the centralized one, we consider the following notion of \emph{relative error}:
\be{Frob_err}
\erel \,\deff\,\frac{\norm{\bY'-\bY}}{\norm{\bY}}.
\ee  
In a sense, $\erel$ measures how much the outcome precision is proportionally compromised by utilizing a distributed protocol while providing privacy/speed up. The entries of the dataset $\bX$ in our experiments are drawn independently from a zero-mean Gaussian distribution with variance $1$. We use $64$ bits for both the fixed-point and the  floating-point numbers to implement both the LCC and the ALCC protocols in our experiments, respectively.

The relative error $\erel$, defined in \eqref{Frob_err}, is computed for the outcome of ALCC in our experiment and is shown in Table\,\ref{FLPvsBeta} for a range of values for the dataset size, that is represented by $m'$, and the Lagrange monomials parameter $\beta$. As discussed in Section\,\ref{privacy_gaussian}, it is expected that increasing $\beta$ results in lower precision outcomes which is also shown in our experiment. But note that it also leads to better privacy as shown in Figure\,\ref{etas}. Also, no notable dependence between the relative error in the outcome of ALCC and the size of dataset $m'$ is observed in Table\,\ref{FLPvsBeta}. This implies that ALCC is scalable with the dataset size as far as the relative error is concerned.

\begin{table}[]
	\centering
	\begin{tabular}{ |c|c|c|c|c|c| } 
		\hline
		\backslashbox{$m'$}  {$\beta$} & $1.1$ & $1.5$& $1.8$ & $2$ \\
		\hline
		$10^4$ & $4.466$ & $3.304$ & $2.316$& $1.699$\\ 
		\hline
		$2\times10^4$& $4.532$ & $3.307$&$2.320$ &$1.713$\\ 
		\hline
		$4\times10^4$& $4.584$ & $3.306$&$2.331$ &$1.723$\\ 
		\hline
		$6\times10^4$& $4.602$ & $3.316$&$2.326$ &$1.727$\\ 
		\hline
		$8\times10^4$& $4.612$ & $3.313$&$2.332$ &$1.731$\\ 
		\hline
		$10^5$& $4.614$ & $3.320$  & $2.334$ &$1.728$\\ 
		\hline
	\end{tabular}
	\caption{\small Demonstration of $-\log_{10} (\erel)$ in ALCC  for multiple dataset sizes and $\beta=1.1, 1.5, 1.8, 2$. Other parameters are $k=5, t=3, s=0, N=15, \sigma_n=10^6, n=100$ for all schemes.}\label{FLPvsBeta}
\end{table}

Next, the performance of LCC \cite{yu2019lagrange} employing fixed-point numbers is compared to that of ALCC from the relative error perspective. In Figure\,\ref{FXPvsFLP}, the relative error is plotted for both LCC and ALCC versus the parameter $m'$, that is proportional to the size of the dataset. For LCC, this is plotted for a few different choices for the size of the underlying finite field $p$, according to the discussion in Section\,\ref{three:B} and keeping in mind that the total number of available bits for representation is $64$. In particular, the first case discussed in Section\,\ref{three:B} is assumed where the worker nodes compute the results module $p$ only once after the polynomial evaluations are completed. Also, for ALCC, $\erel$ is plotted for two values of $\beta$. It can be observed in Figure\,\ref{FXPvsFLP} that for all the scenarios considered for LCC, there exists a certain threshold for $m'$ after which the computation results become very unreliable due to a very high $\erel$. As discussed earlier, this significant precision loss is mostly due to overflow errors that are inherent to the fixed-point implementation employed by LCC. As expected, the sharp increase in $\erel$ occurs at a larger value for $m'$ when a larger $p$ is picked. However, the choice of $p$ is limited by the number of bits available for representing fixed-point numbers. Furthermore, the advantage of ALCC compared to LCC is evident in Figure\,\ref{FXPvsFLP} by observing that the relative error in the outcome of ALCC with floating-point implementation is \emph{almost} constant for the considered range of sizes of the dataset. This motivates employing ALCC in certain problems involving very large datasets. 
\begin{figure}[t]
	\begin{center}
		\hspace{-11mm}\includegraphics[width=.9\linewidth]{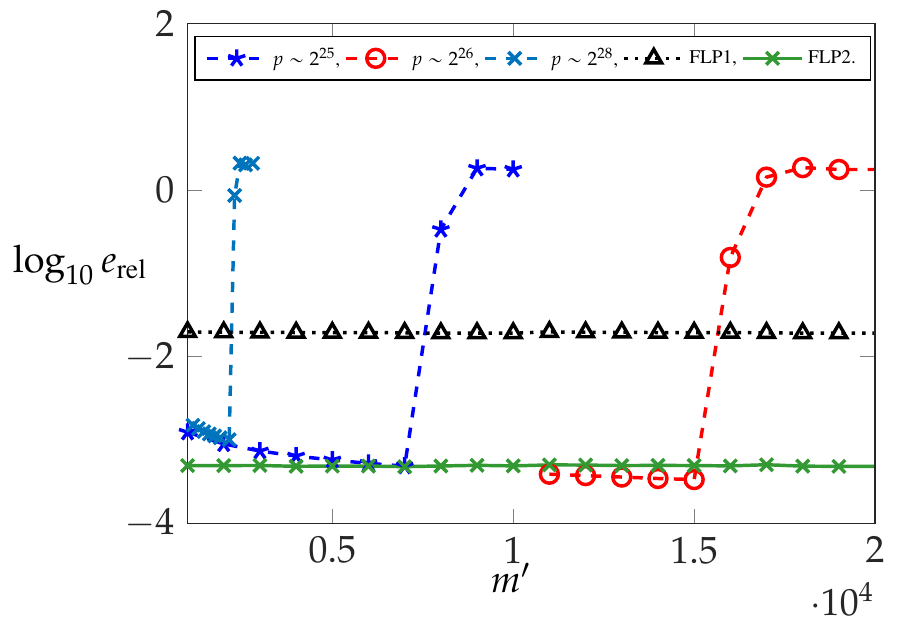}
		\caption{\small Comparison of the relative error in the outcome between ALCC and LCC. For LCC, different values of $p$ are considered ($p \sim 2^{25}$, $p \sim 2^{26}$, $p \sim 2^{28}$). For ALCC, $\beta=2$ and $\beta=1.5$ are considered for FLP1 and FLP2, respectively. Also, in both protocols we have $k=5$, $t=3$, $s=0$, $N=15$, $\sigma_n=10^6$, and $n=100$.}
		\label{FXPvsFLP}
	\end{center}
\end{figure}

\section{Conclusion}\label{sec:conclusion}

In this paper, the Lagrange coded computing framework is extended to the analog domain in order to efficiently evaluate polynomials over real-valued datasets in a distributed fashion. To this end, the analog Lagrange coded computing (ALCC) protocol is proposed that leverages Lagrange polynomials with a certain set of parameters carefully chosen in the complex plane. The privacy of ALCC is measured in terms of the DS and the MIS security metrics in the analog domain. By utilizing the relations between the DS and the MIS security measures and the existing results on the capacity of MIMO channel with correlated noise, bounds on the privacy level of data in ALCC, amidst possible collusion of workers, is characterized in terms of the aforementioned measures. Moreover, the accuracy of the outcome of ALCC is characterized assuming that the floating-point numbers are employed in the implementation of the protocol. Furthermore, a new trade-off between the accuracy of the outcome and the privacy level of the protocol is characterized that is controlled by the choice of Lagrange polynomial parameters. In our experiments, the ALCC is adopted to perform matrix-matrix multiplication and  the outcome is compared to the computation result in a centralized fashion. Finally, the scalability of ALCC and LCC with respect to the dataset size are compared together. It is shown that the accuracy of LCC significantly diminishes after the dataset size passes a certain threshold while the accuracy of ALCC remains almost constant for a wide range of dataset sizes. 

There are several directions for future work. Characterizing the accuracy and the privacy level of the ALCC protocol for a general choice of Lagrange monomial parameters $\beta_j$'s in the complex plane is an interesting direction. In particular, it is not known what choice of  $\beta_j$'s provides the best possible accuracy-privacy trade-off in ALCC and how tight the bounds provided in this paper are with respect to such an optimal scenario. Another direction is to extend ALCC in order to take into account the presence of Byzantine workers, i.e., the worker nodes that deliberately send erroneous computation results \cite{blanchard2017machine, guerraoui2013highly,cramer2015secure,bogdanov2008sharemind}. Providing an efficient and numerically accurate counterpart of Reed-Solomon decoding algorithm in the analog domain would be the main challenge in this direction. Adopting ALLC to provide speed up in performing computational tasks involved in a wide range of applications such as decentralized control, distributed optimization, data mining, etc. \cite{heydariben2018distributed,liu2005random,heydaribeni2019distributed,zhang2018improving,heydaribeni2018distributed} is another future direction. Generalizing ALCC in order to evaluate multiple polynomials in one round by applying techniques utilized in multi-user secret sharing \cite{soleymani2018distributed} is another approach to be considered for future work.     
\bibliographystyle{IEEEtran}
\bibliography{document,ref}

\end{document}